%% file: root.tex
\documentclass[11pt]{article}
\input{preface/includes}

\input{preface/macros}

\title{The Point-to-Set Principle and the Dimensions of Hamel Bases}

\author{Jack H. Lutz\thanks{
Department of Computer Science, Iowa State University, Ames, IA 50011 USA. Email: \href{mailto:lutz@iastate.edu}{lutz@iastate.edu} Research supported in part by National Science Foundation grants 1545028 and 1900716.} 
\and Renrui Qi\thanks{School of Mathematics and Statistics, Victoria University of Wellington, Wellington, New Zealand.  Email: \href{mailto:renrui.qi@sms.vuw.ac.nz}{renrui.qi@sms.vuw.ac.nz}}
\and Liang Yu\thanks{Department of Mathematics, Nanjing University, Jiangsu Province 210093, P.R. of China. Email: \href{mailto:yuliang.nju@gmail.com}{yuliang.nju@gmail.com}
Research supported in part by the National Science Foundation of China, No. 12025103.}
}
\begin{document}
\maketitle

\begin{abstract}
We prove that, for every $0 \leq s \leq 1$, there is a Hamel basis of the vector space of reals over the field of rationals that has Hausdorff dimension $s$.

The logic of our proof is of particular interest.  The statement of our theorem is classical; it does not involve the theory of computing.  However, our proof makes essential use of algorithmic fractal dimension--a computability-theoretic construct--and the point-to-set principle of J. Lutz and N. Lutz (2018).
\end{abstract}

\input{src/section_1.tex}

\input{src/section_2.tex}
\input{src/section_3.tex}

\input{src/section_4.tex}
\input{src/Acknowledgement}

\bibliography{root}

    \newpage

\end{document}

%% file: preface/includes.tex
\usepackage{amsmath}
\usepackage{amssymb}
\usepackage{amsthm}
\usepackage{mathtools}
\usepackage{mathrsfs}
\usepackage{thmtools,thm-restate}

\usepackage{url}
\usepackage{hyperref}
\usepackage{array}
\usepackage{enumerate}

\usepackage[letterpaper,margin=1.2in]{geometry}

\usepackage[parfill]{parskip} % Activate to begin paragraphs with an empty line rather than an indent

%% file: preface/macros.tex
\DeclareMathOperator{\spann}{span}

\newlength{\arrow}

\newcommand*{\N}{\mathbb{N}}
\newcommand*{\R}{\mathbb{R}}
\newcommand*{\Q}{\mathbb{Q}}

\newtheorem{theorem}{Theorem}

\theoremstyle{definition} % amsthm only

%% file: src/section_1.tex
\section{Introduction}
This brief paper is an intellectual export from the theory of computing to classical mathematics, i.e., mathematics not ostensibly involving the theory of computing.  This introduction describes our main theorem and then explains how its proof uses computability theory.

Two fundamental theorems of linear algebra state that every vector space has a basis and that any two bases of a vector space have the same cardinality, which is called the \emph{dimension} of the vector space.  When the vector space has finite dimension, the proofs of these facts are completely constructive and are standard undergraduate fare \cite{apost1969,halmos2017finite}. However, in 1905, Hamel \cite{hamel1905basis} used the axiom of choice to prove that these two theorems hold for all vector spaces.\footnote{In 1984, confirming a 1908 conjecture of Zermelo \cite{zermelo1908}, Blass \cite{blass1984existence} proved that the existence of bases for all vector spaces implies the axiom of choice.} Accordingly, infinite bases of vector spaces are traditionally called \emph{Hamel bases}.

The canonical case of Hamel bases is the vector space $\R$ of real numbers over the field $\Q$ of rational numbers.  This case is the topic of the present paper, so here as in many papers, all ``Hamel bases'' are bases of $\R$ over $\Q$.  Hamel \cite{hamel1905basis} showed that every Hamel basis must have the cardinality of the continuum.  Sierpinski \cite{sierpinski1920question} showed that every Hamel basis has inner Lebesgue measure 0, whence every measurable Hamel basis has measure 0.  Jones \cite{jones1942measure} showed that the Cantor middle-thirds set contains a Hamel basis (which thus has Lebesgue measure 0) and that no Hamel basis is analytic, i.e., $\mathbf{ \Sigma^1_1}$.  Numerous other investigations of Hamel bases have appeared in the literature.

In this paper we investigate the Hausdorff dimensions of Hamel bases.  Our main theorem says that, for every real number $s \in [0,1]$, there is a Hamel basis with Hausdorff dimension exactly $s$.

Although our main theorem says nothing about the theory of computing, our proof of it uses algorithmic fractal dimension, which is a computability-theoretic construct.  Specifically, the \emph{algorithmic dimension} of a number $x \in \R$ is
\begin{equation} 
\dim(x) = \liminf_{n \rightarrow \infty} K(x[0..n-1])/n,
\end{equation}where $x[0..n-1]$ is the string consisting of the first n bits of the binary expansion of $x$, and $K(x[0..n-1])$ is the Kolmogorov complexity of this string \cite{lutz2003dimensions,jMayo02}.  Since $K(x[0..n-1])$ is the algorithmic information content of $x[0..n-1]$ (and is where computability theory comes into the picture \cite{shen2017kolm,downey2010algorithmic,LiVit19}), this says that $dim(x)$ is the {\it lower asymptotic algorithmic information density} of the real number $x$.  This and related algorithmic fractal dimensions are discussed in the recent surveys \cite{LuMay21,LuLu2020}.

The bridge between the above algorithmic dimension and our investigation of the Hausdorff dimensions of Hamel bases is the Point-to-Set Principle of J. Lutz and N. Lutz \cite{LutLut18}.  This principle gives complete pointwise characterizations of various classical fractal dimensions in terms of relativizations of their algorithmic counterparts.  Specialized to the case of Hausdorff dimension in $\R$, this principle says that, for every set $E \subseteq \R$,

\begin{equation}
    \dim_H(E) = \min_{A \subseteq \N} \sup_{x \in E} \dim^A(x).
\end{equation}That is, the \emph{classical} Hausdorff dimension of a set $E$ is the minimum, for all oracles $A$, of the supremum of the algorithmic dimensions, relative to $A$, of the individual points in $E$.

The Point-to-Set Principle is especially useful for proving otherwise difficult lower bounds on the Hausdorff dimensions of sets.  It implies that, to prove a lower bound on the Hausdorff dimension of a set $E$, it suffices to let $A$ be an arbitrary oracle and prove a corresponding lower bound on the algorithmic dimension, relative to $A$, of a judiciously chosen point in $E$.  This ability to reason from the relativized dimensions of single points to the dimensions of entire sets has recently been used to prove several new theorems about classical fractal dimensions.  Examples include stronger lower bounds on the Hausdorff dimensions of generalized Furstenberg sets \cite{LuStu20}; extension of the fractal intersection formula from Borel sets to arbitrary sets \cite{DBLP:journals/toct/Lutz21}; the extension of Marstrand's projection theorem from analytic sets to arbitrary sets, provided that their Hausdorff and packing dimensions coincide \cite{LuStu18}; a proof that, if $V = L$, then the maximal thin co-analytic set has Hausdorff dimension $1$\cite{Slaman}; a further relaxation of the hypothesis of Marstrand's projection theorem \cite{DBLP:conf/stacs/Stull22}; and an improved bound on the Hausdorff dimensions of pinned distance sets \cite{DBLP:journals/corr/abs-2207-12501}.

An encouraging aspect of the above successes of the Point-to-Set Principle is that the proofs do not all look alike.  The principle has turned out to be quite flexible, allowing investigators to invoke it in arguments appropriate to various settings.  In the present paper we use it in the following way.  Given a target dimension $s \in [0,1]$, we construct a Hamel basis $B$ of $\R$ over $\Q$ by a transfinite recursion, putting a single real into $B$ at each stage. At even stages, we add points to $B$ that enable us to use the Point-to-Set Principle to conclude that $B$ has Hausdorff dimension $s$.  At odd stages, we ensure that $B$ spans $\R$.  The details appear in Section 3.

%% file: src/section_2.tex
%!TEX root = ../ARCRN.tex
\section{Preliminaries}

All logarithms here are base-$2$. For $B \subseteq \R$, we write $\spann(B)$ for the linear span of $B$ over $\Q$, i.e., the set of all linear combinations \[ x = \sum_{u \in I} g(u)u,\] where $I \subseteq B$ is finite and $g: I \rightarrow \Q \setminus \{0\}$.

We refer to standard set theory texts, e.g., \cite{halmos2017naive,roitman1990introset,moschovakis2005notes} or early sections of more advanced texts, for background on ordinal numbers and transfinite recursion. We write $\omega$ for the first infinite ordinal and $2^{\aleph_0}$ for the cardinality of the continuum.

For each $s \in (0,1]$, we define a Cantor set $C_s \subseteq [0,1]$ that we use in the proof of our main theorem in Section 3. A set $C_0$ for the case $s=0$ is defined separately below. The definitions and key properties of these sets appear in this preliminary section because they are well understood.

Fix $s\in (0,1]$, and let $r=2^{-\frac{1}{s}}$, noting that $r\in (0,\frac{1}{2}]$. For each $i \in \mathbb{N}$, let 
\[
    r_i= 
\begin{cases}
    r & \text{if } r\in (0, \frac{1}{2})\\
    \frac{1}{2}-2^{-(i+2)}              & \text{if } r=\frac{1}{2},
\end{cases}
\]
noting that $\mathbf{r}=(r_i | i\in \mathbb{N})$ is a sequence in $(0,\frac{1}{2})$ that converges to $r$ in any case.

Define a family $\{I_w | w\in \{0,1\}^*\}$ of closed intervals $I_w \subseteq [0,1]$ by the following recursion. 
\begin{enumerate}[(i)]
    \item $I_\lambda =[0,1]$. (We write $\lambda$ for the empty string.)
    \item If $I_w=[a,b]$, then 
\[I_{w0} = [a, a+r_{|w|}(b-a)]\]
and 
\[I_{w1} = [b-r_{|w|}(b-a),b].\]
\end{enumerate}
For each $k\in \mathbb{N}$, let 
\[C_{s,k} = \bigcup_{w\in \{0,1\}^k}I_w,\]
and let 
\[C_s = \bigcap_{k=0}^\infty C_{s,k}.\]
Note that $C_s$ is a "Cantor middle $1-2r$ set" and that $C_\frac{1}{\log 3}$ is the familiar Cantor middle-thirds set. In any case, $C_s$ is a self-similar fractal whose Hausdorff dimension, by standard methods, is 
\[\dim_H(C_s)=\frac{-1}{\log r}=s.\]
Moreover, each point $x\in C_s$ is specified in the obvious manner by a unique \emph{coding sequence} $T_x \in \{0,1\}^\omega$, and the main theorem of Lutz and Mayodormo\cite{jLutMay08} (relativized to an oracle for the real number $r$) tells us that 
\[\dim^r(x)=s \dim^r(T_x)\]
holds for each $x\in C_s$. In particular, since every sequence $T$ in the set $\text{RAND}^r$ of all sequences that are algorithmically random relative to $r$ has dimension 1 relative to $r$, and since the set $\text{RAND}^r$ has the cardinality of the continuum, it follows that the set 
\[D^r = \{x\in C_s | \dim^r(x) = s\}\]
has the cardinality of the continuum. This entire argument relativizes, so for any oracle $A\subseteq \mathbb{N}$ from which $r$ can be computed, the set 
\[D^A = \{x\in C_s | \dim^A(x) = s\}\]
has the cardinality of the continuum.

The remaining property of $C_s$ that is needed for our argument in section 4 is that 
\[\spann(C_s)=\mathbb{R}.\]
This follows immediately from the main result of Cabrelli, Hare, and Molter \cite{Cabrelli2002SumsOC}, which tells us that there exists $N$ such that the $N$-fold sumset 
\[C_s + ... + C_s = \{x_1 + ... +x_N| \text{each } x_i\in C_s\}\]
contains an interval.

We now turn to the case $s=0$. For each binary sequence ${\mathbf{b} = (b_i | i\in \mathbb{N}) \in \{0,1\}^\omega}$, let 
\[\text{supp}(\mathbf{b}) = \{i\in \mathbb{N} | b_i =1\}\]
be the \emph{support} of $\mathbf{b}$, let 
\[\underline{den}(\mathbf{b})= \liminf_{n \to \infty} \frac{|\text{supp}(\mathbf{b}) \cap \{0,...,n-1\}|}{n} \]
be the \emph{lower asymptotic density} of $\mathbf{b}$, and let 
\[x(\mathbf{b})= \sum_{i=0}^\infty 2^{-(i+1)}b_i\]
be the real number \emph{represented by} $\mathbf{b}$.

Let 
\[C_0=\{x(\mathbf{b})|\mathbf{b}\in \{0,1\}^\omega \text{ and } \underline{den}(b)=0\}.\]
It is well known \cite{lutz2003dimensions} that $\dim(x)=0$ for all $x\in C_0$, whence $\dim_H(C_0)=0$. It is clear that $C_0$ has the cardinality of the continuum.

To see that $\spann(C_0)=\mathbb{R}$, let $x\in [0,1]$. It suffices to show that there exists $x_0,x_1 \in C_0$ such that $x_0+x_1=x$. For each $k\in \mathbb{N}$, let $m_k=k\cdot 2^k$, noting that $m_k=o(m_{k+1}-m_k)$ as $k \to \infty$. Partition $\mathbb{N}$ into the two sets
\[ J_0 = \bigcup_{k=0}^\infty [m_{2k}, m_{2k+1}),\]
\[ J_1 = \bigcup_{k=0}^\infty [m_{2k+1}, m_{2k+2}),\]
where the "intervals" here are the obvious sets of natural numbers. Fix a binary sequence
\[\mathbf{b} = (b_i | i\in \mathbb{N})\in \{0,1\}^\omega\]
such that $x(\mathbf{b})=x$. For each $a\in\{0,1\}$, define the binary sequence
\[\mathbf{b}^{(a)} = (b_i^{(a)} | i\in \mathbb{N})\in \{0,1\}^\omega\]
by
\[
    b_i^{(a)}= 
\begin{cases}
    b_i & \text{if } i\in J_a\\
    0              & \text{if } i\in J_{1-a}
\end{cases}
\]
for all $i\in \mathbb{N}$, and let $x_a=x(\mathbf{b}^{(a)})$. Then $x_0, x_1 \in C_0$ and $x_0+x_1=x$, confirming that $\spann(C_0)=\mathbb{R}$.

In summary, for each $s\in [0,1]$, we have by known methods a set $C_s \subseteq [0,1]$ with the following three properties.
\begin{enumerate}[(i)]
    \item $\dim_H(C_s)=s.$
    \item For all oracles $A\subseteq \mathbb{N}$ that compute $s$, the set 
    \[D^A=\{x\in C_s | \dim^A(x)=s\}\]
    has the cardinality of the continuum.
    \item $\spann(C_s)=\mathbb{R}$.
\end{enumerate}

%% file: src/section_3.tex
%!TEX root = ../ARCRN.tex
\section{Dimensions of Hamel Bases}
In this section we prove our main theorem, which is the following.

\begin{theorem}
\label{Thm3.1}
For every $s \in [0,1]$ there is a Hamel basis B of $\R$ over $\Q$ such that $\dim_H(B) = s$.
\end{theorem}
\begin{proof}

Let $s\in [0,1]$, and define the set $C_s$ as in Section 2. Let 
\[(x_\alpha | \alpha < 2^{\aleph_0})\]
be a wellordering of $C_s$, and let 
\[((A_\beta, y_\beta)| \beta < 2^{\aleph_0})\]
be a wellordering of the set
\[D = \{(A,y) \in \mathcal{P}(\mathbb{N}) \times C_s | s \leq_T A \text{ and }\dim^A(y)=s\}.\]
Define the sequence
\[(u_\gamma | \gamma < 2^{\aleph_0})\]
of real numbers by the following transfinite recursion. Given $\gamma < 2^{\aleph_0}$, let ${B_\gamma = \{u_\delta| \delta < \gamma\}}$. Write $\gamma = \xi + k$, where $\xi$ is 0 or a limit ordinal and $k\in \mathbb{N}$. Call $\gamma$ \emph{even} or \emph{odd} if $k$ is even or odd, respectively.
\begin{enumerate}[(i)]
    \item If $\gamma = \xi + 2k$ is even, let 
    \[u_\gamma = y_\beta,\]
    where $\beta$ is the least ordinal such that $A_\beta = A_{\xi + k}$ and $y_\beta \not \in \spann(B_\gamma)$.
    \item If $\gamma$ is odd, let
    \[u_\gamma = x_\alpha\]
    be the first element of $C_s \setminus \spann(B_\gamma)$.
\end{enumerate}

To see that this recursion is well defined, let $\gamma < 2^{\aleph_0}$. Then $|B_\gamma|< 2^{\aleph_0}$, so $|\spann(B_\gamma)| < 2^{\aleph_0}$. Since $|D|=|C_s|= 2^{\aleph_0}$, it follows that $u_\gamma$ exists.

Let 
\[B=\{u_\gamma | \gamma < 2^{\aleph_0} \}.\]
The rest of this proof establishes that $B$ is a Hamel basis of $\mathbb{R}$ over $\mathbb{Q}$ and that $\dim_H(B)=s$.

We first show that $\spann(B)=\mathbb{R}$. For this, since we saw in Section 2 that $\spann(C_s)=\mathbb{R}$, it suffices to show that $C_s \subseteq \spann(B)$, i.e., that the set ${E=C_s \setminus \spann(B)}$ is empty. To this end, let $x$ be a lower bound of $E$ in our wellordering of $C_s$. That is, let $x=x_\alpha$, where $\alpha < 2^{\aleph_0}$ and $x_{\alpha'} \in \spann(B)$ holds for all $\alpha' < \alpha$. Since $\alpha < 2^{\aleph_0} = |B|$ and only finitely many elements of $B$ are required to put each $\alpha'$ into $\spann(B)$, it follows that there exists $\gamma < 2^{\aleph_0}$ such that 
\[\{x_{\alpha'} | \alpha' < \alpha \} \subseteq \spann(B_\gamma).\]
Moreover, we can insist here that $\gamma$ be odd. Then $u_\gamma = x_\alpha = x$, so ${x\in B \subseteq \spann(B)}$, so $x\not \in E$. We have now shown that no lower bound of $E$ is an element of $E$, i.e., that $E$ has no least element in our wellordering. This implies that $E= \emptyset$, completing our proof that $\spann(B) = \mathbb{R}$.

Since our construction ensures that $u_\gamma \not \in \spann(B_r)$ holds for all $\gamma < 2^{\aleph_0}$, we now have that $B$ is a Hamel basis of $\mathbb{R}$ over $\mathbb{Q}$.

We finish this proof by showing that $\dim_H(B)=s$. It is clear that ${\dim_H(B) \leq s}$, since $B \subseteq C_s$ and $\dim_H(C_s)=s$. Hence it suffices to show that 
\[ s \in (0,1] \Rightarrow \dim_H(B) \geq s.\]
For this, assume that $s\in (0,1]$, and let $A\subseteq \mathbb{N}$ be an oracle such that $s \leq_T A$. We saw in Section 2 that the set 
\[D^A = \{y \in C_s | \dim^A(y)=s\}\] 
has the cardinality of the continuum, so there exist ordinals $\gamma$ and $\beta$ such that $A_\beta = A$ and $y_\beta = u_\gamma \in B$. This implies that $\dim^A(y_\beta)=s$. Writing $u(A) = y_\beta$ here, it follows by the Point-to-Set principle (together with the obvious fact that $\dim^C(x) \leq \dim^D(x)$ holds whenever $D \leq_T C)$ that 
\begin{align*}
        \dim_H(B) &= \min_{A \subseteq \N} \sup_{u\in B} \dim^A(u) \\
        &=\min_{\substack{A \subseteq \N \\ s\leq_T A}} \sup_{u\in B} \dim^A(u)\\
        &\geq \min_{\substack{A \subseteq \N \\ s\leq_T A}} \dim^A(u(A))\\
        &= s.
\end{align*}
%\qedr
\end{proof}
%\end{apthm}

%% file: src/section_4.tex
\section{Conclusion}
Orponen \cite{orp2021} has already found classical proofs of two projection theorems of N. Lutz and Stull \cite{LuStu18} that were first proven via the Point-to-Set Principle, and he generalized these theorems in the process.  We conjecture that our theorem on Hamel bases also admits a classical proof.  More generally, we look forward to a better understanding of the power and limitations of computability-theoretic methods for discovering proofs of new theorems in classical mathematics.

%% file: src/Acknowledgement.tex
\section*{Acknowledgment}
We thank two anonymous referees for useful suggestions.

%% file: root.bbl
\begin{thebibliography}{10}

\bibitem{apost1969}
Tom~M. Apostol.
\newblock {\em Calculus, Volume II: Multi-variable Calculus and Linear Algebra, with Applications to Differential Equations and Probability}.
\newblock John Wiley \& Sons, 1969.

\bibitem{blass1984existence}
Andreas Blass.
\newblock Existence of bases implies the axiom of choice.
\newblock {\em Contemporary Mathematics}, 31, 1984.

\bibitem{Cabrelli2002SumsOC}
Carlos Cabrelli, Kathryn~E. Hare, and Ursula Molter.
\newblock Sums of {C}antor sets yielding an interval.
\newblock {\em Journal of the Australian Mathematical Society}, 73:405 -- 418, 2002.
\newblock URL: \url{https://api.semanticscholar.org/CorpusID:11396262}.

\bibitem{downey2010algorithmic}
Rodney~G. Downey and Denis~R. Hirschfeldt.
\newblock {\em Algorithmic Randomness and Complexity}.
\newblock Springer Science \& Business Media, 2010.

\bibitem{halmos2017finite}
Paul~R. Halmos.
\newblock {\em Finite-dimensional Vector Spaces}.
\newblock Dover, 2017.

\bibitem{halmos2017naive}
Paul~R. Halmos.
\newblock {\em Naive Set Theory}.
\newblock Dover, 2017.

\bibitem{hamel1905basis}
Georg Hamel.
\newblock Eine {B}asis aller {Z}ahlen und die unstetigen {L}{\"o}sungen der {F}unktionalgleichung: f (x+ y)= f (x)+ f (y).
\newblock {\em Mathematische Annalen}, 60(3):459--462, 1905.

\bibitem{jones1942measure}
F.B. Jones.
\newblock Measure and other properties of a {H}amel basis.
\newblock {\em Bulletin of the American Mathematical Society}, 48(6):472--481, 1942.

\bibitem{LiVit19}
Ming Li and Paul~M.B. Vitányi.
\newblock {\em An Introduction to Kolmogorov Complexity and Its Applications}.
\newblock Springer Publishing Company, Incorporated, 4 edition, 2019.

\bibitem{lutz2003dimensions}
Jack~H. Lutz.
\newblock The dimensions of individual strings and sequences.
\newblock {\em Information and Computation}, 187(1):49--79, 2003.

\bibitem{LutLut18}
Jack~H. Lutz and Neil Lutz.
\newblock Algorithmic information, plane {K}akeya sets, and conditional dimension.
\newblock {\em ACM Transactions on Computation Theory (TOCT)}, 10(2):1--22, 2018.

\bibitem{LuLu2020}
Jack~H. Lutz and Neil Lutz.
\newblock Who asked us? how the theory of computing answers questions about analysis.
\newblock In {\em Complexity and Approximation}, pages 48--56. Springer, 2020.

\bibitem{jLutMay08}
Jack~H. Lutz and Elvira Mayordomo.
\newblock Dimensions of points in self-similar fractals.
\newblock {\em SIAM Journal on Computing}, 38(3):1080--1112, 2008.
\newblock \href {https://doi.org/10.1137/070684689} {\path{doi:10.1137/070684689}}.

\bibitem{LuMay21}
Jack~H. Lutz and Elvira Mayordomo.
\newblock Algorithmic fractal dimensions in geometric measure theory.
\newblock In {\em Handbook of Computability and Complexity in Analysis}, pages 271--302. Springer, 2021.

\bibitem{DBLP:journals/toct/Lutz21}
Neil Lutz.
\newblock Fractal intersections and products via algorithmic dimension.
\newblock {\em {ACM} Trans. Comput. Theory}, 13(3):14:1--14:15, 2021.
\newblock \href {https://doi.org/10.1145/3460948} {\path{doi:10.1145/3460948}}.

\bibitem{LuStu18}
Neil Lutz and Donald~M. Stull.
\newblock Projection theorems using effective dimension.
\newblock In {\em 43rd International Symposium on Mathematical Foundations of Computer Science (MFCS 2018)}. Schloss Dagstuhl-Leibniz-Zentrum fuer Informatik, 2018.

\bibitem{LuStu20}
Neil Lutz and Donald~M. Stull.
\newblock Bounding the dimension of points on a line.
\newblock {\em Information and Computation}, 275:104601, 2020.

\bibitem{jMayo02}
Elvira Mayordomo.
\newblock A {K}olmogorov complexity characterization of constructive {H}ausdorff dimension.
\newblock {\em Information Processing Letters}, 84(1):1--3, 2002.

\bibitem{moschovakis2005notes}
Yiannis Moschovakis.
\newblock {\em Notes on Set Theory}.
\newblock Springer, 2005.

\bibitem{orp2021}
Tuomas Orponen.
\newblock Combinatorial proofs of two theorems of {L}utz and {S}tull.
\newblock In {\em Mathematical Proceedings of the Cambridge Philosophical Society}, pages 1--12. Cambridge University Press, 2021.

\bibitem{roitman1990introset}
Judith Roitman.
\newblock {\em Introduction to Modern Set Theory}.
\newblock John Wiley \& Sons, 1990.

\bibitem{shen2017kolm}
Alexander Shen, Vladimir~A. Uspensky, and Nikolay Vereshchagin.
\newblock {\em Kolmogorov Complexity and Algorithmic Randomness}, volume 220.
\newblock American Mathematical Soc., 2017.

\bibitem{sierpinski1920question}
Wac{\l}aw Sierpi{\'n}ski.
\newblock Sur la question de la mesurabilit{\'e} de la base de {M}. {H}amel.
\newblock {\em Fundamenta Mathematicae}, 1(1):105--111, 1920.

\bibitem{Slaman}
Theodore~A. Slaman.
\newblock On capacitability for co-analytic sets.
\newblock {\em New Zealand J. Math.}, 52:865--869, 2021 [2021--2022].
\newblock \href {https://doi.org/10.1007/s13226-021-00101-z} {\path{doi:10.1007/s13226-021-00101-z}}.

\bibitem{DBLP:conf/stacs/Stull22}
Donald~M. Stull.
\newblock Optimal oracles for point-to-set principles.
\newblock In Petra Berenbrink and Benjamin Monmege, editors, {\em 39th International Symposium on Theoretical Aspects of Computer Science, {STACS} 2022, March 15-18, 2022, Marseille, France (Virtual Conference)}, volume 219 of {\em LIPIcs}, pages 57:1--57:17. Schloss Dagstuhl - Leibniz-Zentrum f{\"{u}}r Informatik, 2022.
\newblock \href {https://doi.org/10.4230/LIPIcs.STACS.2022.57} {\path{doi:10.4230/LIPIcs.STACS.2022.57}}.

\bibitem{DBLP:journals/corr/abs-2207-12501}
Donald~M. Stull.
\newblock Pinned distance sets using effective dimension.
\newblock {\em CoRR}, abs/2207.12501, 2022.
\newblock \href {https://arxiv.org/abs/2207.12501} {\path{arXiv:2207.12501}}, \href {https://doi.org/10.48550/arXiv.2207.12501} {\path{doi:10.48550/arXiv.2207.12501}}.

\bibitem{zermelo1908}
Ernst Zermelo.
\newblock Untersuchungen {\"u}ber die {G}rundlagen der {M}engenlehre. i.
\newblock {\em Mathematische Annalen}, 65(2):261--281, 1908.

\end{thebibliography}
